\documentclass[conference]{IEEEtran}
\ifCLASSINFOpdf
\else
\fi
\usepackage{graphics} 
\usepackage{mathptmx} 
\usepackage{times} 
\usepackage{amsmath} 
\usepackage{amssymb}  
\usepackage{tikz}
\usepackage{algorithm}
\usepackage{algpseudocode}
\usepackage{array}
\usepackage{caption}
\usepackage{times}
\usepackage{todonotes}
\usepackage[normalem]{ulem}

\begin{document}
%

\newif\ifclientA
\newtheorem{lemma}{Lemma}
\newtheorem{corollary}{Corollary}
\newtheorem{theorem}{Theorem}
\newtheorem{proof}{Proof}
\newtheorem{assumption}{Assumption}
\newtheorem{remark}{Remark}
\newtheorem{definition}{Definition}	
\newtheorem{properties}{Property}

\clientAtrue 

\title{Path Planning with Moving Obstacles Using Stochastic Optimal Control}

\author{\IEEEauthorblockN{Seyyed Reza Jafari}
\IEEEauthorblockA{Division of Automatic Control\\Department of Electrical Engineering\\
Linköping University\\
Linköping, Sweden\\
Email: seyyed.reza.jafari@liu.se}
\and
\IEEEauthorblockN{Anders Hansson}
\IEEEauthorblockA{Division of Automatic Control\\Department of Electrical Engineering\\
Linköping University\\
Linköping, Sweden\\
Email: anders.g.hansson@liu.se}
\and
\IEEEauthorblockN{Bo Wahlberg}
\IEEEauthorblockA{Division of Decision and Control Systems\\
KTH Royal Institute of Technology\\
Stockholm, Sweden\\
Email: bo@kth.se}}


%


\maketitle

\begin{abstract}
Navigating a collision-free and optimal trajectory for a robot is a challenging task, particularly in environments with moving obstacles such as humans.
We formulate this problem as a stochastic optimal control problem.
Since solving the full problem is computationally demanding, we introduce a tractable approximation whose Bellman equation can be solved efficiently.
The resulting value function is then incorporated as a terminal penalty in an online rollout framework.
We construct a trade-off curve between safety and performance to identify an appropriate weighting between them, and compare the performance with other methods.
Simulation results show that the proposed rollout approach can be tuned to reach the target in nearly the same expected time as receding horizon $A^\star$ while maintaining a larger expected minimum distance to the moving obstacle. The results also show that the proposed method outperforms the considered CBF-based methods when a larger obstacle clearance is desired, while achieving comparable performance otherwise.

\end{abstract}


%
\IEEEpeerreviewmaketitle

\section{Introduction}
There are numerous applications of autonomous mobile robots that are discussed in the literature. In general, these applications are divided into two groups: indoor applications, such as delivering packages, cargo, and cleaning large buildings~\cite{c2}, and outdoor field robotics application~\cite{c15}. 
Finding an optimal path for the robot to reach its destination is a crucial task in these applications.
Path planning is one of the primary challenges that must be solved before mobile robots can autonomously navigate and explore complex environments~\cite{c14}.

Based on the nature of the environment that the robot is working in, the mobile path planning is divided into two classes, i.e., static environment and dynamic environment~\cite{c2},~\cite{c11}. In the static environment, the obstacles are static and do not move. In the dynamic environment, the obstacles move during the planning task, while the robot is moving to reach its destination. 
The primary goal of path planning is to ensure safe, efficient, and collision-free navigation in both static and dynamic environments~\cite{c9}.
Path planning algorithms can be broadly categorized into several groups. The first group is classical algorithms, such as Dijkstra's algorithm. The second group is heuristic-based approaches, e.g., the $A^\star$ algorithm. The third group is artificial intelligence-based approaches, such as reinforcement learning, fuzzy logic, and so on. The fourth group is hybrid-based approaches, which combine elements from all the above classes of path planning~\cite{c2}. However, this is not the only way to classify path planning algorithms. In~\cite{c11},~\cite{c10}, and~\cite{c12}, other classifications for different approaches to path planning are also presented. 

In this paper, we consider a case in which there is a robot and a stochastically moving obstacle in the environment. Also, we assume that there is a static target, and we want to control the robot to reach this target. 
We formulate the problem of finding an optimal path as an infinite horizon stochastic optimal control problem. The associated Bellman equation can be solved using value iteration. Notably, for many problems, achieving an exact solution is simply infeasible due to what Bellman properly termed the ‘curse of dimensionality’ in~\cite[p.~94]{c4}. This problem refers to the exponential increase in the required computations as the problem size grows~\cite[p.~282]{c5}.
In order to circumvent this problem, we exploit geometric symmetry and propose to use~\cite{c21} to reduce the dimension of the domain of the value function. We then solve the resulting reduced Bellman equation using fitted value iteration. 
Fitted value iteration is a parametric approximation method in which we parameterize the value function as a regression model such that it closely satisfies the Bellman equation~\cite[p.~211]{c3}. 

To obtain a policy in real time, we employ the rollout approach.
In this method, the control sequence over a finite horizon is optimized online while the cost-to-go of a base policy serves as a terminal penalty~\cite{c23,c27}, and~\cite[Chapter 4]{c28}.
At each time step, only the first control input of the optimized sequence is applied, and the procedure is repeated at the next step, resulting in a receding horizon implementation that improves upon the base policy. 

Path planning is inherently a multi-objective task~\cite{c25}. In particular, two common and often competing objectives are safety and performance~\cite{c24}.
In this paper, we introduce two cost terms: one associated with safety and the other with performance, and seek an appropriate trade-off between these objectives.
An important feature of the proposed formulation is that this trade-off can be tuned explicitly through a single parameter $\lambda$, which gives a simple mechanism for selecting a policy with a desired balance between performance and safety.

We compare our method with a stochastic discrete time CBF approach~\cite{c22} and a receding horizon $A^\star$ method. A key advantage of the proposed method is that the parameter $\lambda$ provides a direct and interpretable way to tune the trade-off between time to target and obstacle clearance. This makes it straightforward to generate a trade-off curve and to select a desired operating point. In contrast, for CBF-based methods, performance depends strongly on the choice of barrier function and associated tuning parameters, and obtaining a desired trade-off is less direct. Our numerical results show that the proposed method can achieve nearly the same expected time to target as receding horizon $A^\star$ while maintaining a larger expected minimum distance to the obstacle. They also show that the proposed method outperforms the considered CBF-based methods when a larger obstacle clearance is desired, while providing comparable performance otherwise.

Summarizing, the main contributions of this paper are:
\begin{enumerate}
    \item We formulate an optimal path planning problem as a stochastic optimal control problem and develop an efficient approximate solution suitable for online rollout.
  
    \item We use geometric symmetry to reduce the domain of the value function, enabling its representation with a reduced number of variables, which significantly lowers the computational cost of solving the Bellman equation.

    \item We show through simulations that the proposed method provides an easily tunable trade-off between time to target and obstacle clearance, while comparing favorably with receding horizon $A^\star$ and the considered CBF-based methods.

\end{enumerate}

\section{Mathematical Model} \label{Sec_II}
Consider a problem where there is a robot, a stochastically moving obstacle, and a static target in $2$-dimensional space. The goal for the robot is to reach the target as fast as possible without colliding with the obstacle. 
In this section, we define the mathematical model of the dynamics and formulate an optimal control problem that will meet the objectives.

\subsection{System Dynamics}
Let $r_k=\begin{bmatrix}r_k^x & r_k^y\end{bmatrix}^T\in \mathbb{R}^2$ and $h_k=\begin{bmatrix}h_k^x & h_k^y\end{bmatrix}^T \in \mathbb{R}^2$ represent the positions of the robot and the moving obstacle at time $k$, respectively. Moreover, assume that $t=\begin{bmatrix}t^x & t^y\end{bmatrix}^T \in \mathbb{R}^2$ is the target position. 
The dynamics of the robot and the moving obstacle are described by the following equations:
\begin{equation}
    \begin{cases}
        r_{k+1}=r_k+u_k \triangleq F_r\left(r_k,u_k\right) \\
        h_{k+1}=h_k+w_k \triangleq F_h\left(h_k,w_k\right)
    \end{cases}
    \label{eq1}
\end{equation}
where $F_r:\mathbb{R}^2 \times \mathcal{U} \rightarrow \mathbb{R}^2$, $F_h:\mathbb{R}^2 \times \mathcal{W} \rightarrow \mathbb{R}^2$. The sets $\mathcal{U}$ and $\mathcal{W}$ represent all possible movements that the the robot, and the moving obstacle can take, respectively.
Here $\mathcal{U}=\left\{R_u\begin{bmatrix}\cos\left(\alpha\right) & \sin\left(\alpha\right)\end{bmatrix}^T \in \mathbb{R}^2: R_u \in \{0,1\}, \alpha \in \mathcal{A}_1\right\}$, and $\mathcal{W}=\left\{R_w\begin{bmatrix}\cos\left(\gamma\right) & \sin\left(\gamma\right)\end{bmatrix}^T \in \mathbb{R}^2: R_w \in \{0,1\}, \gamma \in \mathcal{A}_2\right\}$, where $\mathcal{A}_i = \{\frac{q\pi}{n_i} \mid q=0,1,\cdots,2n_i-1\}$, and $n_i,i=1,2$ are natural numbers. We can rewrite the set $\mathcal{W}$ as $\{w^0,w^1,\cdots,w^{2n_2}\}$, in which $w^{i}=\begin{bmatrix}\cos\left(\frac{i\pi}{n_2}\right)& \sin\left(\frac{i\pi}{n_2}\right)\end{bmatrix}^T$, for $0 \leq i \leq 2n_2-1$, and $w^{i}=\begin{bmatrix}0 & 0\end{bmatrix}^T$ for $i=2n_2$. We let $w$ take different values within the set $\mathcal{W}$ with different probabilities.
The models defined in (\ref{eq1}) are known as \textit{constant position} models \cite{c19}.

\begin{assumption} \label{assum1}
    We assume that the robot and the moving obstacle are solid bodies with radius of $\frac{R}{2}$. This assumption simplifies the problem, but we can also handle the general case where the robot and obstacle have different sizes. Moreover, when the distance between the robot and the moving obstacle is less than or equal to $R$, i.e., $d \triangleq ||h-r||_2 \leq R$, the robot has collided with the moving obstacle. Additionally, when the distance between the robot and the target is less than or equal to $R$, i.e., $e \triangleq ||r-t||_2 \leq R$, the robot has reached the target.
\end{assumption}

\begin{assumption} \label{assum2}
    We assume that $w$ is radially symmetric, i.e., its probability density function depends only on the magnitude $R_w$ and is constant for all orientations $\gamma$.
\end{assumption}

\subsection{Stochastic Infinite Horizon Optimal Control Problem}
Assume that $r_0$ and $h_0$ are initial positions of the robot and the moving obstacle, respectively. 
The problem is to find an optimal policy that enables the robot to reach the final state, i.e., the target, while trying to avoid colliding with the moving obstacle. The objectives are twofold: (i) minimize the expected time to reach the target, and (ii) maximize the minimum distance maintained between the robot and the moving obstacle. Obtaining the optimal solution to this problem is non-trivial. We formulate it as a stochastic optimal control problem with a specifically designed incremental cost. This incremental cost is expressed as a linear combination of two terms: the first term encourages the robot to reach the target as quickly as possible, while the second term penalizes situations in which the robot comes too close to the moving obstacle.
The mathematical description of this problem is:
\begin{equation}
    \begin{aligned}
        &\operatorname{minimize} \lim_{N \to \infty} \mathbb{E} \left[\sum_{k=0}^{N-1} f(h_k,r_k) \right]  \\
        &\text{subject to} \quad  
        \begin{cases}
            r_{k+1}= F_r\left(r_k,u_k\right) \\
            h_{k+1}=F_h\left(h_k,w_k\right)
       \end{cases}, \quad k \geq 0 \\
        &\text{\phantom{subject to}} \quad u_k \in \mathcal{U}, w_k \in \mathcal{W} \quad k \geq 0 
    \end{aligned}
    \label{eq2}
\end{equation}
 where $\mathbb{E}[\cdot]$ denotes  mathematical expectation, and where $f:\mathbb{R}^2 \times \mathbb{R}^2 \rightarrow \mathbb{R}_{+}$ is the so-called incremental cost. Here, $\mathbb{R}_+$ is the set of nonnegative real numbers. We define the function $f$ for a given target position $t$ as
\begin{equation}
        f(h,r) = \begin{cases}
        0, \quad ||r-t||_2 \leq R\\
        \lambda (||r-t||_2-R)^2 + \frac{1-\lambda}{||h-r||_2+\epsilon}, \quad ||r-t||_2 > R
        \end{cases}  \\
    \label{eq3}
\end{equation}
where $\lambda \in [0,1]$ is a tuning parameter, and $\epsilon \in \mathbb{R}_{++}$ is a small number. Here, $\mathbb{R}_{++}$ is the set of positive real numbers. 
The motivation behind selecting this function is that when the robot reaches the target, i.e., $||r-t||_2 \leq R$, the incurred cost is zero. This makes the target position $t$ a so-called cost-free state. In other words, once the robot reaches the target, it will stay there without incurring any cost. This cost free state is well-known in the literature of Stochastic Shortest Path (SSP) problem~\cite{c16,c28} and is also referred to as the destination~\cite{c1}. Moreover, when the robot gets close to the moving obstacle or is far from the target, the incurred cost is high. 

\subsection{Solution of Optimal Control}
In this section, we will solve the infinite horizon optimal control problem (\ref{eq2}) using the Bellman equation:
\begin{equation}
    V(h,r) = \min_{u \in \mathcal{U}} \mathbb{E}\left[f(h,r)+V\left(F_h(h,w),F_r(r,u)\right)\right]
    \label{eq4}
\end{equation}
If we find a solution $V$ to this equation, the optimal policy is the minimizing argument in the right-hand side of the equation. Thus the optimal policy is a function of $(h,r)$.
Note that, the Bellman equation in (\ref{eq4}) must be solved for all $(h,r)$ and for a given target position $t$.
This problem leads to the 'curse of dimensionality' due to the large state space.

\section{Geometric Symmetry}
In this section, we utilize the symmetry of the value function to reduce its domain and solve the Bellman equation in a lower dimensional state space. 
\begin{definition} \label{def1}
    We define $e \triangleq ||r-t||_2$ as the distance between the robot and the target, $d \triangleq ||h-r||_2$ as the distance between the robot and the moving obstacle, $\theta \triangleq \arccos\left(\frac{(r-t)^T(h-r)}{||h-r||_2 \cdot||r-t||_2}\right)$ as the angle between these two vectors $r-t$ and $h-r$.
\end{definition}

By defining these new variables, we can rewrite the incremental cost in~\eqref{eq3} as
\begin{equation}
    f(h,r) =\begin{cases}
        0, \quad e \leq R\\
         \lambda (e-R)^2 + (1-\lambda)\frac{1}{d+\epsilon}, \quad e > R
    \end{cases}
    \triangleq \Bar{f}(d,e)
    \label{eq4b}
\end{equation}
where $\Bar{f} : \mathbb{R}_+ \times \mathbb{R}_+ \rightarrow \mathbb{R}_+$.
\begin{definition}
    Consider the two pairs, $(h_1,r_1)$ and $(h_2,r_2)$, in Fig.~\ref{Fig1}. If $d_1=d_2$, $e_1=e_2$, and $\theta_1 =\pm \theta_2$, then these pairs are said to be symmetric around $t$ .
\end{definition}

\begin{figure}[h]
      \centering
      \begin{tikzpicture}[scale=0.45]

    
    \draw[->] (-4,-3.5) -- (7.5+0.5,-3.5) node[right] {$x$};
    \draw[->] (-3.5,-4) -- (-3.5,7.5+0.5) node[above] {$y$};
    \coordinate (t) at (2,2);
    \coordinate (r_1) at (3,{2+sqrt(8)});
    \coordinate (r_2) at ({2+sqrt(8)},3);
    \coordinate (h_1) at (3.2,7);
    \coordinate (h_2) at ({7},{2.95});
    
    \fill[black] (t) circle (3pt) node[above left] {$\mathbf{t}$};

    \fill[black] (r_1) circle (3pt) node[above left] {$\mathbf{r_1}$};

    \fill[black] (r_2) circle (3pt) node[above left] {$\mathbf{r_2}$};

    \fill[black] (h_1) circle (3pt) node[above left] {$\mathbf{h_1}$};

    \fill[black] (h_2) circle (3pt) node[right] {$\mathbf{h_2}$};
    
    \draw[->, black] (t) -- (r_1) node[midway, above left] {$\mathbf{e_1}$};
    \path (t) -- (r_1) coordinate[pos=0.5] (midpoint);
    \draw[dashed,-] (2,2) -- (5,2) node[right] {};

    \draw[->, black] (t) -- (r_2) node[midway, above] {$\mathbf{e_2}$};
    \path (t) -- (r_2) coordinate[pos=0.5] (midpoint);

    \draw[->, black] (r_1) -- (h_1) node[midway, above left] {$\mathbf{d_1}$};
    \path (r_1) -- (h_1) coordinate[pos=0.5] (midpoint);
    \draw[dashed,black] (r_1) -- (4,{2+2*sqrt(8)}) node[right] {};

    \draw[->, black] (r_2) -- (h_2) node[midway, below] {$\mathbf{d_2}$};
    \path (r_2) -- (h_2) coordinate[pos=0.5] (midpoint);
    \draw[dashed,black] (r_2) -- (7,{2+5/sqrt(8)}) node[right] {};
   
    \foreach \x in {1,2,3,4,5,6,7,8,9} \draw[draw=none] (\x,0) -- (\x,-0.1);
    \foreach \y in {1,2,3,4,5,6,7,8,9} \draw[draw=none] (0,\y) -- (-0.1,\y);

    \draw[->] (r_2) ++(18:1) arc (35:-7:0.5) node[midway, right] {\textcolor{black}{$\theta_2$}};

    \draw[->] (r_1) ++(70:1) arc (52:75:0.7) node[midway, right] {\textcolor{black}{$\theta_1$}};

    \draw[dashed, black] (t) circle (3);
    \draw[dashed, black] (t) circle (5.12);
    \draw[dashed, black] (r_1) circle (2.2);
    \draw[dashed, black] (r_2) circle (2.2);
\end{tikzpicture}
      \caption{ Geometrical representation of position of robot, and moving obstacle for two symmetric pairs $(h_1,r_1)$, and $(h_2,r_2)$}
      \label{Fig1}
\end{figure}

The paper~\cite{c21} introduces a framework that exploits symmetry to achieve dimensionality reduction in dynamic programming formulations of finite-horizon stochastic optimal control problems. In these settings, the inherent symmetry is shown to be preserved by both the optimal cost-to-go function and the corresponding optimal policy.
It can be shown, using the result of~\cite{c21}, that the value function inherits this symmetry for the case of the stochastic shortest path problem under suitable assumptions. In other words, $V(h_1,r_1)=V(h_2,r_2)$ for all symmetric pairs $(h_1,r_1)$ and $(h_2,r_2)$. 
\ifclientA 
See Appendix A for the full proof.
\else 
See Appendix for the full proof.
\fi

This reduces the domain of the value function from $4$ variables $\begin{bmatrix}h & r\end{bmatrix}^T \in \mathbb{R}^4$ to $3$ scalar variables $d,e$, and $\theta$. 
\ifclientA 
See Appendix B for the full derivation of the reduced state formulation.
\else 
See~\cite[Appendix B]{c33} for the full derivation of the reduced state formulation.
\fi

Using this symmetry property, we evaluate the Bellman equation only on a cross-section of symmetric states defined by
\begin{equation}
    \mathcal{C}=\left\{x=\begin{bmatrix}h & r\end{bmatrix}^T \in \mathbb{R}^4: r^y=t^y,r^x>t^x\right\}
    \label{eq5a}
\end{equation}
Hence, the Bellman equation in reduced coordinates becomes:
\begin{equation}
    \begin{aligned}
         & W(d,e,\theta)= \min_{u \in \mathcal{U}} \mathbb{E}\left[\Bar{f}(d,e)+W\left(\Bar{F}(d,e,\theta,u,w)\right)\right]
     \end{aligned}
     \label{eq5}
\end{equation}
where $\Bar{F}:\mathbb{R}_+ \times \mathbb{R}_+ \times [0,\pi] \times \mathcal{U} \times \mathcal{W} \rightarrow \mathbb{R}_+ \times \mathbb{R}_+ \times [0,\pi]$ is the system dynamics on $\mathcal{C}$. 
To define the function $\Bar{F}$, let $(h,r)$ be a pair defined by $(d,e,\theta)$ on $\mathcal{C}$ as in Fig.~\ref{Fig2}, i.e.,
\begin{equation}
    \begin{cases}
         h = r + d \begin{bmatrix}\cos(\theta) & \sin(\theta)\end{bmatrix}^T \\
         r = t + e \begin{bmatrix}1 & 0\end{bmatrix}^T
    \end{cases}
    \label{eq7}
\end{equation}
Then we can express the next values of $\left(d,e,\theta\right)$ as functions of the current values of $\left(d,e,\theta\right)$ and $u$, and $w$. More precisely, we have
\begin{equation}
    \begin{cases}
        e_{+} = ||e\begin{bmatrix}1 & 0\end{bmatrix}^T +u||_2 \triangleq F_e(e,u) \\
        \theta_{+} = \arccos\left(\frac{\nu^T \xi}{||\nu||_2 ||\xi||_2}\right) \triangleq F_\theta(d,e,\theta,u,w) \\
        d_{+} = ||d \begin{bmatrix}\cos(\theta) & \sin(\theta)\end{bmatrix}^T +w-u||_2 \triangleq F_d(d,\theta,u,w)
    \end{cases}
    \label{eq8}
\end{equation}
where $F_e: \mathbb{R}_+ \times \mathcal{U} \rightarrow \mathbb{R}_+$, $F_\theta: \mathbb{R}_+ \times \mathbb{R}_+ \times [0,\pi] \times \mathcal{U} \times \mathcal{W} \rightarrow [0,\pi]$, and $F_d: \mathbb{R}_+ \times [0,\pi] \times \mathcal{U} \times \mathcal{W} \rightarrow \mathbb{R}_+$, and $\Bar{F}= \begin{bmatrix}F_d&F_e&F_\theta\end{bmatrix}^T$, and where
\begin{equation}
    \nu = \begin{bmatrix}e + u^x \\ u^y\end{bmatrix}, \quad
    \xi = \begin{bmatrix}d \cos(\theta) + w^x -u^x \\ d \sin(\theta) + w^y - u^y\end{bmatrix}
    \label{eqq11}
\end{equation}
\ifclientA 
See Appendix C for the full derivation of the reduced-state dynamics.
\else 
See~\cite[Appendix C]{c33} for the full derivation of the reduced-state dynamics.
\fi
\begin{figure}[!b]
      \centering
      \begin{tikzpicture}[xscale=0.7, yscale=0.7]
     \clip (0.5,1) rectangle (7.5,6.5);
    \draw[->] (0.9,1.2) -- (7,1.2) node[right] {$x$};
   \draw[->] (1.2,0.9) -- (1.2,5.5) node[above] {$y$};

    \coordinate (t_1) at (2,2);
    \coordinate (r_1) at (4.5,2);
    \coordinate (r_2) at ({2+sqrt(8)},3);
    \coordinate (h_1) at ({6},{2.5});
    \coordinate (h_2) at ({6.7},{5});
    
    \fill[black] (t_1) circle (3pt) node[above left] {$\mathbf{t}$};
    
    \fill[black] (r_1) circle (3pt) node[above left] {$\mathbf{r}$};
    
    \fill[black] (h_1) circle (3pt) node[above right] {$\mathbf{h}$};
    
    \fill[black] (r_2) circle (3pt) node[above left] {$\mathbf{r}_+$};

    \fill[black] (h_2) circle (3pt) node[above left] {$\mathbf{h}_+$};

    \draw[->, black] (t_1) -- (r_2) node[midway, above] {$\mathbf{e}_+$};
    \draw[dashed,-] (2,2) -- (6,2) node[right] {};
    \draw[dashed,-] (r_2) -- (7,3) node[right] {};
    \draw[->, black] (t_1) -- (r_1) node[midway, below] {$\mathbf{e}$};
    \draw[->, black] (r_2) -- (h_2) node[midway, above left] {$\mathbf{d}_+$};
    \draw[->, black] (r_1) -- (h_1) node[midway, above] {$\mathbf{d}$};
    \path (r_2) -- (h_2) coordinate[pos=0.5] (midpoint);
    \draw[dashed,black] (r_2) -- (7,{2+5/sqrt(8)}) node[right] {};
    \draw[->] (4.5,2) ++(0:1) arc (0:25:0.8) node[midway, below right] {\textcolor{black}{$\theta$}};
    \draw[->] ({2.1+sqrt(8)},{3+0.2/sqrt(8)}) ++(15:1.2) arc (30:55:1.5) node[right] {\textcolor{black}{$\theta_+$}};

    \foreach \x in {1,2,3,4,5,6,7,8,9} \draw[draw=none] (\x,0) -- (\x,-0.1);
    \foreach \y in {1,2,3,4,5,6,7} \draw[draw=none] (0,\y) -- (-0.1,\y);

\end{tikzpicture}
      \caption{ Geometrical representation of position of robot, obstacle, and target with future values}
      \label{Fig2}
\end{figure}
We can use value iterations as in~\cite{c3} to solve the Bellman equation in~(\ref{eq5}):
\begin{equation}
    W_{k+1}(d,e,\theta)= \min_{u \in \mathcal{U}} \mathbb{E}\left[\Bar{f}(d,e)+W_k\left(\Bar{F}(d,e,\theta,u,w)\right)\right]
    \label{eq6}
\end{equation}
In the next section, we use fitted value iteration for this recursion. 

\section{fitted Value Iteration} \label{Sec_IV}
In fitted value iteration, we approximate the value function with a parameterized function. We select a set of sample states and adjust the parameters so that the Bellman equation is closely satisfied at those samples. The parameters are then updated iteratively, and this process repeats until they converge.

We introduce a vector-valued function $\varphi: \mathbb{R}_+ \times \mathbb{R}_+ \times [0,\pi] \rightarrow \mathbb{R}^p$, known as a feature vector, and parameters $a_k \in \mathbb{R}^p$. Our goal is to approximate the function $W_k(d,e,\theta)$, with $\Tilde{W}: \mathbb{R}_+ \times \mathbb{R}_+ \times [0,\pi] \times \mathbb{R}^p \rightarrow \mathbb{R}$. One possibility is to use a linear regression model as follows:
\begin{equation}
    \Tilde{W}(d,e,\theta,a_k)=a_k^T\varphi(d,e,\theta)
    \label{eq13}
\end{equation}
We consider a piecewise constant approximation of $W_k$ by assuming that the $i$-th component of $\varphi$, $ i \in \{1,2,\cdots,p\}$ is defined as
\begin{equation}
     \varphi_i(d,e,\theta)=\begin{cases}
        1,\quad (d,e,\theta) \in \mathcal{S}_i \\
        0, \quad (d,e,\theta) \notin \mathcal{S}_i
    \end{cases}
    \label{eq14}
\end{equation}
where $\mathcal{S}_i$ is a partition of $\mathbb{R}_+ \times \mathbb{R}_+ \times \left[0,\pi \right]$. More precisely, write $\mathcal{D} = \cup_{j=1}^{N_d-1} \mathcal{D}_j = [d_1,d_{N_d}]$, where $\mathcal{D}_j=[d_j,d_{j+1}]$, such that $\mathcal{D}_j \cap \mathcal{D}_{\Bar{j}}= \emptyset$, for $j \neq \Bar{j}$.
Similarly, we define $\mathcal{E}= \cup_{l=1}^{N_e-1} \mathcal{E}_l$, and $\Theta = \cup_{m=1}^{N_\theta-1} \Theta_m$. 
Then we let $\mathcal{S}_i=\mathcal{D}_j \times \mathcal{E}_l \times \Theta_m$, where $j \in \left\{1,2,\cdots,N_d-1\right\}$, $l \in \left\{1,2,\cdots,N_e-1\right\}$, $m \in \left\{1,2,\cdots,N_\theta-1\right\}$, and $i=(j-1)(N_e-1)(N_\theta-1)+(l-1)(N_\theta-1)+m$.
This results in $p=(N_d-1)(N_e-1)(N_\theta-1)$ partitions.
We refer the reader to~\cite{c3} for more details about fitted value iteration.

We consider $N_r$ different samples $(d^s,e^s,\theta^s) \in \mathbb{R}_+ \times \mathbb{R}_+ \times [0,\pi]$ with at least one sample in each set of the partition. Then, at each iteration, we perform two updates:
\begin{enumerate}
    \item \textbf{Value Update}:
    \begin{equation}
        \beta_k^s=\min_u \mathbb{E} \left[\Bar{f}(d^s,e^s)+\Tilde{W}\left(d_+^s,e_+^s,\theta_+^s,a_k\right)\right] 
        \label{eq15}
    \end{equation}
    \item \textbf{Parameter Update}:
    \begin{equation}
        a_{k+1} = \arg\min_a \sum_{s=1}^{N_r}\left(a^T\varphi(d^s,e^s,\theta^s)-\beta_k^s\right)^2
        \label{eq16}
    \end{equation}
\end{enumerate}
where $d_+^s$, $e_+^s$, and $\theta_+^s$ are defined as in~(\ref{eq8}).
We iterate until $\delta_k \triangleq ||a_{k+1}-a_k||_\infty$ becomes sufficiently small. 
Algorithm~\ref{alg:EFVI} summarizes the fitted value iteration. 
\begin{algorithm}[H]
\caption{Reduced-Space Fitted Value Iteration}
\label{alg:EFVI}
\begin{algorithmic}[1]
\State \textbf{Input:} Incremental Cost $\Bar{f}$, Functions $F_d$, $F_e$ and $F_\theta$, Integer $n_2$, Control Set $\mathcal{U}$, Set $\mathcal{W}$, Tolerance $\epsilon_{tol}>0$, Maximum Number of Iterations $K$, $N_r$ Different Sample Points $(d^s,e^s,\theta^s)$
\State \textbf{Output:} $a$
\State \textbf{Initialize:} $k=0$, $a_0 = 0$, Generate $N_r$ Sample Points
\While{$||a_{k+1}-a_{k}||_{\infty}>\epsilon_{\text{tol}}$ \& $k<K$}

\State \textbf{Value Update}:
    \For{$s \leftarrow 0$ \textbf{to} $N_r$}
                \State $Q(u) \leftarrow 0$
                \For{$u \in \mathcal{U}$}
                    \For{$i \leftarrow 0$ \textbf{to} $2n_2$}
                        \State Set: $w=w^{i}$
                        \State Compute: $d_+,e_+,\theta_+$
                        \State Compute:
                        \State $\Tilde{W}(d_+^s,e_+^s,\theta_+^s,a_k)=a_k^T\varphi(d_+^s,e_+^s,\theta_+^s)$
                        \State Update:
                        \State $Q(u) \leftarrow Q(u) + \frac{1}{2n_2+1}\Tilde{W}(d_{+}^s,e_{+}^s,\theta_{+}^s,a_k)$
                    \EndFor
                    \State $Q(u) \leftarrow Q(u) + \Bar{f}(d^s,e^s)$
                \EndFor  
                \State $\beta_k^s=\min_{u \in \Bar{\mathcal{U}}} Q(u)$
            \EndFor
            \State \textbf{Parameter Update}:
            \State $a_{k+1} = \arg\min_{a} \sum_{s=1}^{N_r}\left(a^T\varphi(d^s,e^s,\theta^s)-\beta_k^s\right)^2$
            \State $k \leftarrow k+1$
            
\EndWhile
\end{algorithmic}
\end{algorithm}

\section{Effect of Probability Distribution and Constraints}

In the proposed method, a crucial assumption that must be satisfied is Assumption~\ref{assum2}.
In this section, we relax this condition and also consider constraints on the states. Notice that under this assumption, we can use symmetry reduction and reduce the dimension of the value function. However, this assumption imposes a restriction on the probability distribution of the motion of the obstacle. In realistic scenarios, such a restriction is not practical, and therefore the assumption must be relaxed to make the proposed approach applicable to real problems. 

We use the optimal value function $V$ that satisfies the Bellman equation in~(\ref{eq4}) and formulate the following $N$-step lookahead rollout problem~\cite{c23}:
\begin{equation}
    \begin{aligned}
        & \underset{u_l \in \mathcal{U}}{\operatorname{minimize}} \quad \mathbb{E} \left[V(h_{k+N},r_{k+
        N}) + \sum_{l=k}^{N+k-1}f(h_k,r_k) \right]  \\
        &\text{subject to} \quad  
        \begin{cases}
            r_{l+1}= F_r\left(r_l,u_l\right) \\
            h_{l+1}=F_h\left(h_l,w_l\right)
       \end{cases}, \;  l=k,k+1,\cdots, N+k-1   \\ 
        &\text{\phantom{subject to}} \quad u_l \in \mathcal{U}, \; w_l \in \mathcal{W},\; x_l=(r_l,h_l) \in \mathcal{X} \quad l \geq 0 
    \end{aligned}
    \label{eq17}
\end{equation}
where $\mathcal{X}$ is the set of all constraints on $x_k$. The expectation is taken with respect to $w$, and the probability distribution of $w$ does not necessarily satisfy Assumption~\ref{assum2}.
Note that for given $h$, $r$, and a static target position $t$, $V(h,r)=W(d,e,\theta)$, where $d,e$, and $\theta$ are defined in Definition~\ref{def1}.

At each time step, we solve the optimization problem in~\eqref{eq17} in real time to obtain the optimal control sequence $u_l^\star$. The first control input of this sequence is then applied, and the remaining elements are discarded.
Had there been no constraints on $x_k$ and the probability distribution had satisfied our previous assumptions, then this online procedure would provide the optimal feedback for~\eqref{eq2}. 
However, when state constraints are included or Assumption~\ref{assum2} is violated, the approach provides a suboptimal solution.
This approach reduces the impact of the curse of dimensionality at the price of a suboptimal solution.

Solving the optimization problem~\eqref{eq17} is computationally intensive. To reduce this computational burden, we apply the \textit{certainty equivalence (CE)} principle, in which the stochastic disturbances $w_k$ are replaced by their expected values~\cite{c23}. Using this approximation, we formulate the following problem
\begin{equation}
    \begin{aligned}
        & \underset{u_l \in \mathcal{U}}{\operatorname{minimize}} \quad \left[V(h_{k+N},r_{k+
        N}) + \sum_{l=k}^{N+k-1}f(h_k,r_k) \right]  \\
        &\text{subject to} \quad  
        \begin{cases}
            r_{l+1}= F_r\left(r_l,u_l\right) \\
            h_{l+1}=F_h\left(h_l,\Bar{w}_l\right)
       \end{cases}, \; l=k,k+1,\cdots, N+k-1   \\ 
        &\text{\phantom{subject to}} \quad u_l \in \mathcal{U}, \; w_l \in \mathcal{W},\; x_l=(r_l,h_l) \in \mathcal{X} \quad l \geq 0 
    \end{aligned}
    \label{eq18}
\end{equation}
where $\Bar{w}_l=\mathbb{E}\left[w_l\right]$.

Next, let $V_{N,\lambda}^{\star}(h_0,r_0,t)$ and $V_{N\text{(CE)},\lambda}^{\star}(h_0,r_0,t)$ denote the optimal values of the optimization problems in~\eqref{eq17} and~\eqref{eq18}, respectively, for a fixed target position $t$, an initial condition $\begin{bmatrix} h_0 & r_0 \end{bmatrix}^\mathsf{T}$, and a given tuning parameter $\lambda$. Recall that the stage cost is a combination of two competing costs, weighted by the parameter $\lambda$ as in~\eqref{eq3}. To evaluate the effect of $\lambda$, we compute the expected performance over a representative set of initial states and targets. For each value of $\lambda$, we define
\begin{equation}
    \begin{aligned}
        & \Bar{V}_{N,\lambda} = \mathbb{E}\left[V^*_{N,\lambda}(h_0,r_0,t)\right] \\
        & \Bar{V}_{N\text{(CE)},\lambda} = \mathbb{E}\left[V_{N\text{(CE)},\lambda}^*(h_0,r_0,t)\right]
    \end{aligned}
    \label{eq19}
\end{equation}
where the expectation is taken over a grid of sampled initial conditions and target positions. The purpose of this evaluation is to study the trade-off between the two performance criteria that form the cost in~\eqref{eq3}, namely the expected time to
reach the target and the expected minimum distance maintained from the moving obstacle. By varying $\lambda$, we obtain the trade-off curve between the two performance criteria.

\section{Numerical Results}

In this section, we consider a numerical example and we use the proposed method described in Algorithm~\ref{alg:EFVI} to find the value function for different values of $\lambda$. 

We let $\!N_d=115,\;N_e=85,\;N_\theta=26\!$, and $\!\{d_1,d_2,\cdots,d_{61}\}=\{0,0.05,0.1,\cdots,3\}\!$, $\!\{d_{62},d_{63},\cdots,d_{115}\}=\{3.5,4,4.5,\cdots,30\}\!$. Also we consider $\!\{e_1,e_2,\cdots,e_{31}\}=\{0,0.1,0.2,\cdots,3\}\!$, $\!\{e_{32},e_{33},\cdots,e_{85}\}=\{3.5,4,4.5,\cdots,30\}\!$. Moreover we let $\!\{\theta_1,\theta_2,\cdots,\theta_{26}\}=\{0,\frac{\pi}{25},\frac{2\pi}{25},\cdots,\pi\}\!$. This results in $p= 239,400$ partitions.
We consider $N_r=718,200$ sample points, i.e., three sample points per partition. Additionally, we let $R=1$, $\epsilon=10^{-8}$ , $\epsilon_{\text{tol}}=10^{-5}$, and $K=20$. The sets $\mathcal{U}$ and $\mathcal{W}$ are defined in Section~\ref{Sec_II}, and we let $n_1=n_2=16$. Moreover, we define the set $\mathcal{X}$ as follows
 \begin{displaymath}
     \mathcal{X} =\left\{\begin{bmatrix}h^x &h^y &r^x&r^y\end{bmatrix}^T \in \mathbb{R}^4: 0 \leq r^x,r^y,h^x,h^y \leq 20 \right\}
 \end{displaymath}
We implement Algorithm~\ref{alg:EFVI} in MATLAB and run it over a set of weighting parameters $\lambda \in [0,1]$.
We define a discrete probability distribution on $\mathcal{W}$ by assigning the unnormalized weights
\begin{displaymath}
    \Tilde{p}^j=\begin{cases}100, \quad \text{if } w^{i_x}>0, \text{ and } w^{i_y}>0 \\
    1, \quad \text{otherwise}\end{cases}
\end{displaymath}
and normalizing them as $ P(w^i)=\frac{\Tilde{p}^i}{\sum_{j=1}^{2n_2+1}\Tilde{p}^j}$.

For the CBF approach~\cite{c22}, at each time instance, we solve the following expectation-based optimization problem
\begin{displaymath}
    \begin{aligned}
        & u_{\text{CBF}} = \arg\min_{u \in \mathcal{U}} ||u-u_{\text{nom}}||_2^2 \\
        &\text{subject to} \quad \mathbb{E}\left[B\left(F_h(h_k,w_k),F_r(r_k,u)\right)\mid x_k\right]\geq \alpha B(x_k)
    \end{aligned}
\end{displaymath}
where $B:\mathbb{R}^4 \rightarrow \mathbb{R}$ is the barrier function defined as $B(x_k)=||h_k-r_k||-d_0$, and $\alpha \in(0,1)$ is a design parameter.
We also consider the CE variant
\begin{displaymath}
    \begin{aligned}
        & u_{\text{CBF~(CE)}} = \arg\min_{u \in \mathcal{U}} ||u-u_{\text{nom}}||_2^2 \\
        &\text{subject to} \quad B\left(F_h(h_k,\Bar{w}_k),F_r(r_k,u)\right)\geq \alpha B(x_k)
    \end{aligned}
\end{displaymath}
where $\Bar{w}_k=\mathbb{E}[w_k]$. 
The nominal controller $u_{\text{nom}}$ is defined as the controller that drives the robot toward the goal as fast as possible without considering the obstacle.

We perform $50$ independent simulations, where in each trial the target position is sampled uniformly at random from the grid. The initial positions of the robot and the moving obstacle are also sampled uniformly at random, subject to the constraints $||r_0-t||_2>1$ and $||h_0-r_0||_2>1$. For each set of initial conditions, we generate $10$ independent realizations of the stochastic disturbance $w$ to approximate the costs in~\eqref{eq17} and~\eqref{eq18}, i.e., $V_{N,\lambda}^{\star}(h_0,r_0,t)$ and $V_{N\text{(CE)},\lambda}^{\star}(h_0,r_0,t)$. For the CBF method, we evaluate the performance for multiple values of $\alpha \in (0,1)$, and $d_0$. 
Fig.~\ref{Fig3} shows the trade-off curve obtained by varying the weighting parameter $\lambda$ compared with the receding horizon $A^\star$, CBF, and CBF (CE) methods. 
Recall that we maximize the expected minimum distance between the robot and the obstacle.
To obtain a standard Pareto trade-off plot, where both axes represent minimization, we display its negative value in Fig.~\ref{Fig3}. Thus, points closer to the origin correspond to better performance in both objectives.
From Fig.~\ref{Fig3}, the proposed method provides a favorable trade-off between expected time to target and expected minimum obstacle distance over the grid of initial conditions. In particular, for operating points with larger desired obstacle clearance, the proposed method achieves better trade-offs than the considered CBF and CBF~(CE) methods. For smaller desired clearance, the performance is comparable. Moreover, the figure shows that the proposed method can be tuned to operate close to the $A^\star$ solution in expected time to target while maintaining a larger expected minimum distance to the obstacle.
It is possible that a different barrier function and tuning could improve the CBF trade-off. However, identifying such a barrier function is not straightforward. In contrast, the proposed method offers a simple and direct tuning mechanism through $\lambda$, which makes it easy to explore and select the desired time-safety trade-off.

Moreover, Fig.~\ref{Fig3} shows that increasing the horizon shifts the trade-off curve further down and to the left, meaning that the longer horizon yields a better trade-off between the objectives (lower values for both costs) compared to curves that lie above and to the right. Furthermore, from this figure, the receding horizon $A^\star$ approach yields a solution that prioritizes minimizing the expected time to reach the target, as its points lie toward the lower-right part of the plot.
\begin{figure}[t]
      \centering
      \includegraphics[scale=0.6]{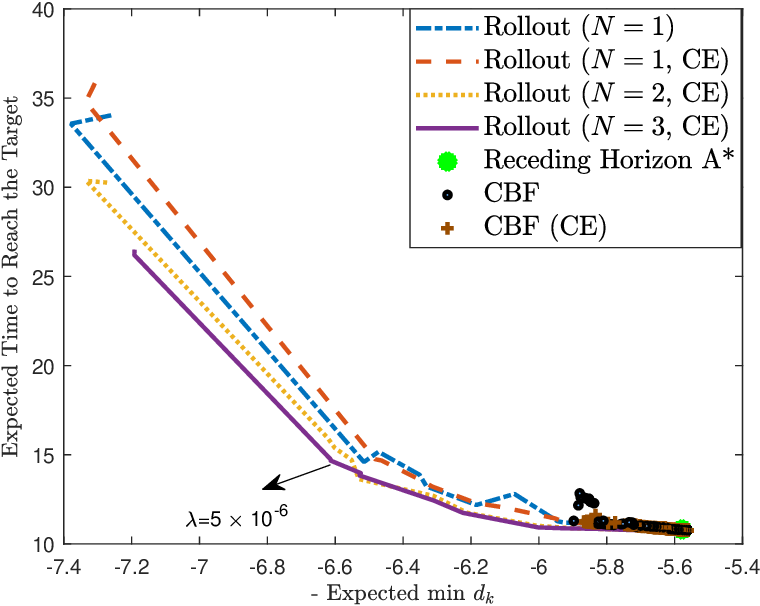}
      \caption{Trade-off curve between expected time to target (horizontal axis) and negative minimum distance to the obstacle (vertical axis) for different horizons, compared with receding horizon $A^\star$, CBF, and CBF~(CE) over the grid}
      \label{Fig3}
\end{figure}

Now, we let $h_0=\begin{bmatrix}2 & 6\end{bmatrix}^T$, $r_0=\begin{bmatrix}4& 12\end{bmatrix}^T$, and $t=\begin{bmatrix}4 & 3\end{bmatrix}^T$.
We consider $100$ different realizations of $w$, and compare the proposed method with the receding-horizon $A^\star$, CBF, and CBF~(CE) approaches. Fig.~\ref{Fig4} illustrates the trade-off curve obtained by varying the weighting parameter $\lambda$ compared with the other methods for the given initial condition and for different combinations of $\alpha \in (0,1)$, and $d_0$. 
From Fig.~\ref{Fig4}, the best trade-off between the competing costs occurs at $\lambda^*=5 \times 10^{-6}$, which is the same value that we obtained from Fig.~\ref{Fig3}.
Note that Fig.~\ref{Fig3} shows the trade-off curve computed over the entire grid, whereas Fig.~\ref{Fig4} illustrates the trade-off curve corresponding to the given initial conditions.
From Fig.~\ref{Fig4}, for the given initial condition, the proposed method outperforms CBF and CBF~(CE) in the region where a larger minimum obstacle distance is desired, while achieving comparable performance in the rest of the trade-off range. Thus, the proposed method is particularly advantageous when safety, in the sense of larger obstacle clearance, is prioritized.

\begin{figure}[h]
      \centering
      \includegraphics[scale=0.6]{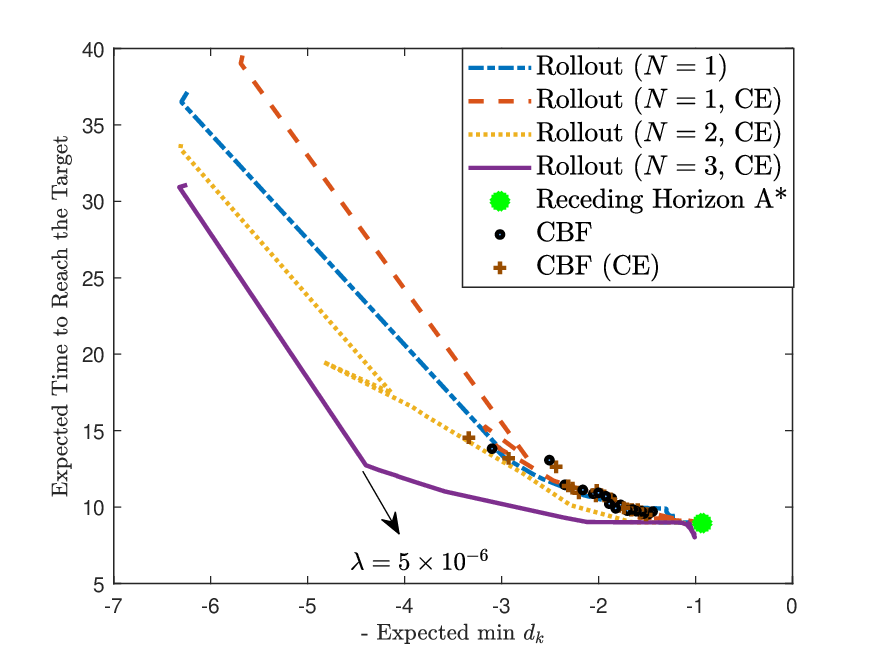}
      \caption{Trade-off curve between expected time to target (horizontal axis) and negative minimum distance to the obstacle (vertical axis) for different horizons, compared with receding-horizon $A^\star$, CBF, and CBF (CE) for the given initial conditions}
      \label{Fig4}
\end{figure}
From Fig.~\ref{Fig4}, we select $\lambda = \lambda^\ast = 5 \times 10^{-6}$ for our method, and $\alpha = 0.75$ and $d_0 = 1$ for the CBF and CBF~(CE) methods, as these values provide the best trade-off for the given initial condition. We then plot the trajectory and control signal for a sample run of the different methods.
 Fig.~\ref{Fig5} shows the trajectory of the robot and the moving obstacle for the different methods. From this figure, the trajectory of the robot when using the proposed approach makes a detour to avoid the moving obstacle. The trajectories produced by the CBF and CBF~(CE) methods are identical in this example, and they also take a detour in order to satisfy the barrier function constraint. In contrast, the receding-horizon $A^\star$ approach does not account for the moving obstacle and moves directly toward the goal without considering it.
Fig.~\ref{Fig6} shows the distance between the robot and the target, and the distance between the robot and the moving obstacle, over time for the different methods. 
This figure shows that receding-horizon $A^\star$ reaches the target slightly faster than the proposed method. However, the proposed method maintains a larger distance from the moving obstacle during the motion. Hence, the proposed method can achieve time to target performance close to $A^\star$ while offering improved safety. Relative to CBF and CBF (CE), the proposed method achieves comparable performance in this example, while offering the practical advantage that its time-safety trade-off can be tuned directly through  $\lambda$.

\begin{figure}[!b]
      \centering
      \includegraphics[scale=0.6]{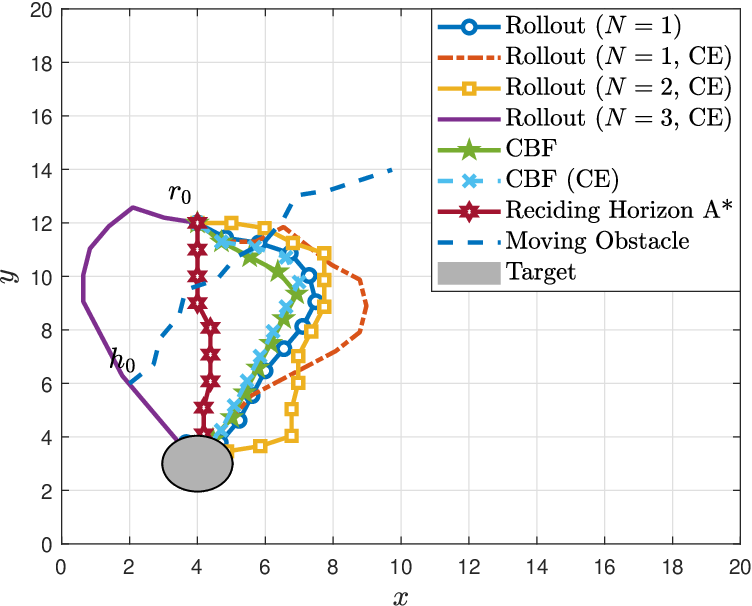}
      \caption{Trajectory of Robot, and Moving Obstacle using Different Methods}
      \label{Fig5}
\end{figure}

\section{Conclusions}

In this paper, we proposed a symmetry reduced stochastic optimal control approach for path planning in the presence of a stochastically moving obstacle. Compared with receding horizon $A^\star$, CBF, and CBF~(CE), the proposed method provides a simple and interpretable tuning mechanism through $\lambda$ to select and visualize the trade-off between time to target and obstacle clearance. The numerical results show that the proposed method can achieve nearly the same expected time to target as receding horizon $A^\star$ while maintaining a larger expected minimum distance to the obstacle. They also show that the proposed method outperforms the considered CBF based methods when a larger obstacle clearance is desired, while offering comparable performance otherwise.
Future work includes extending this framework to three-dimensional environments or to scenarios with moving targets. In the three-dimensional case with a static target, the symmetry reduction decreases the state dimension from $6$ to $3$. Additional directions include applying the reduced dimension value function within a reinforcement learning framework to learn the optimal value function in unknown environments, and generalizing the approach to multiple moving obstacles.

\begin{figure}[!t]
     \centering
      \includegraphics[scale=0.6]{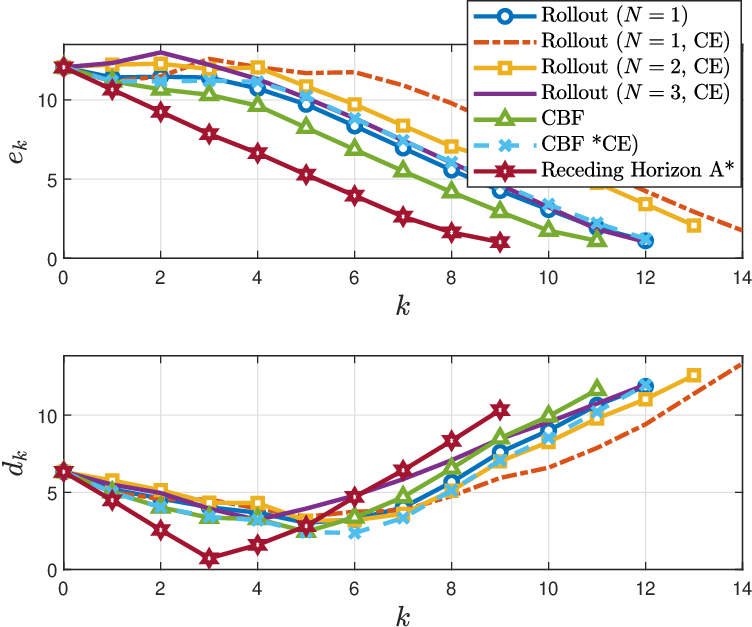}
      \caption{Distance between Robot and Target, and Distance between Robot and Moving obstacle}
      \label{Fig6}
\end{figure}

\section*{APPENDIX}

\ifclientA 
In this appendix, we extend the symmetry–reduction framework of~\cite{c21}, originally developed for finite–horizon stochastic optimal control, to the SSP setting. We then show how this framework can be applied to reduce the dimensionality of the state space in the proposed path planning problem. 

\subsection{Symmetry of Optimal Value Function in Stochastic Shortest Path Problems} \label{App_A}

Let $G$ be a \emph{group}, with group operation $*$.
Consider a bijective function $\phi:G \times \mathcal{X} \rightarrow \mathcal{X}$, which is a so-called \emph{group action}, i.e.,
\begin{subequations}
    \begin{align}
        &\phi\left(\alpha*\beta,x\right) = \phi\left(\alpha,\phi(\beta,x)\right), \quad \forall \alpha,\beta \in G, \forall x \in \mathcal{X} \label{eq_A0a} \\
         &\phi\left(e,x\right)=x \label{eq_A0b}
    \end{align}
    \label{eq_A0}
\end{subequations}
where $e$ is the identity element of $G$.
We consider an SSP problem~\cite{c1} and show that, under suitable invariance assumptions on the system dynamics, cost, and noise distribution, the optimal value function $V^\star$ is \emph{invariant}; i.e., for any $\alpha \in G$
\begin{equation}
    V^\star\left(\phi\left(\alpha,x\right)\right) = V^\star\left(x\right),
   \qquad \forall x\in \mathcal{X}
   \label{eq_A1}
\end{equation}

To prove invariance of the optimal value function, we make the following assumptions:

\begin{enumerate}
    \item \textbf{System dynamics and cost}:
    The system evolves according to
   \begin{equation}
       x_+ = F(x,u,w)
       \label{eq_A2}
   \end{equation}
   with an incremental cost $f(x,u,w) \geq 0$. Where $F: \mathcal{X} \times \mathcal{U} \times \mathcal{W} \rightarrow \mathcal{X}$, and $f:\mathcal{X} \times \mathcal{U} \times \mathcal{W} \rightarrow \mathbb{R}_+$.
    Also we assume that there is a cost-free absorbing terminal state $t \in \mathcal{X}$ satisfying
   \begin{equation}
       F(t,u,w) = t, \quad f(t,u,w)=0, \quad  \forall u \in \mathcal{U}, w \in \mathcal{W}
       \label{eq_A3}
   \end{equation}

   \item \textbf{Bellman Operator}: The Bellman operator $T$ is defined by
   \begin{equation}
       (TV)(x) = 
       \begin{cases}
           0, & x = t, \\
           \displaystyle \min_{u\in \mathcal{U}} \, \mathbb{E}\big[\, f(x,u,w) + V(F(x,u,w)) \,\big], & x \neq t.
          \end{cases}
       \label{eq_A8}
    \end{equation}
   We assume that there exists at least one proper policy from the initial states—i.e., a policy that reaches $t$ with probability one and finite expected cost. Under this assumption, the Bellman operator admits a unique minimal nonnegative fixed point,~\cite[Proposition~1]{c1}:
   \begin{equation}
       V^\star = TV^\star,
       \qquad
        V^\star \le \tilde{V} \quad \text{for any other fixed point } \tilde{V}
       \label{eq_A3b}
   \end{equation}
   Thus, $V^\star$ is the optimal value function of the SSP problem.

   \item \textbf{Group invariance}:
   In addition to the bijective group action $\phi$, we assume that there exist bijective group actions $\chi:G \times\mathcal{U}\to \mathcal{U}$ and $\psi:G \times \mathcal{W}\to \mathcal{W}$ acting on $u$ and $w$, respectively; such that
   \begin{equation}
       \begin{aligned}
      F(\phi\left(\alpha,x\right), \chi\left(\alpha,u\right), \psi\left(\alpha,w\right)) &= \phi\!\left(\alpha,F(x,u,w)\right)  \\
      f(\phi(\alpha,x), \chi(\alpha,u), \psi(\alpha,w)) &= f(x,u,w)
       \end{aligned}
       \label{eq_A4}
   \end{equation}
   \item \textbf{Noise distribution invariance}:
   We assume that $w$ takes values in a finite set $\mathcal W$ with probabilities $P(w)$, such that:
   \begin{equation}
       P(\psi(\alpha,w)) = P(w)
       \label{eq_A5}
   \end{equation}


   \item \textbf{Terminal Invariance}:
    The terminal state is fixed under the group action $\phi$
    \begin{equation}
        \phi(\alpha,t) = t
        \label{eq_A5b}
    \end{equation}
\end{enumerate}


\begin{theorem} \label{The1}
Under the above assumptions, the optimal value function $V^\star$ is invariant under the group action and satisfies~\eqref{eq_A1}.

\end{theorem}


\begin{proof}
First, we show that
\begin{equation}
     (TV) \left(\phi(\alpha,x)\right) \;=\; \left(T\left(V\left(\phi\left(\alpha,x\right)\right)\right)\right)(x)
   \qquad \forall\; V:\mathcal{X} \to \mathbb{R}_+
   \label{eq_A7}
\end{equation}
The proof proceeds in two cases:
\begin{enumerate}
    \item \textbf{Case 1}, $x=t$:

    From terminal invariance~\eqref{eq_A5b}, both sides of~\eqref{eq_A7} equal zero.
     \item \textbf{Case 2}, $x \neq t$:

     We begin by evaluating the left-hand side
     \begin{equation}
         \begin{aligned}
             & (TV)\big(\phi(\alpha,x)\big) \\
           & = \min_{u \in \mathcal{U}} \, \mathbb{E}_w\big[\, f(\phi(\alpha,x),u,w) + V(F(\phi(\alpha,x),u,w)) \,\big]
         \end{aligned}
         \label{eq_A12}
     \end{equation}
     Using the change of variables $u = \chi(\alpha,\tilde u)$, and $w=\psi(\alpha,\tilde{w})$,  and the invariance properties~\eqref{eq_A4}–\eqref{eq_A5}, we obtain
    \begin{equation}
    \small
    \begin{aligned}
        &(TV)\big(\phi(\alpha,x)\big) \\
         & = \min_{u \in \mathcal{U}} \, \sum_w P(w)\big[\, f(\phi(\alpha,x),u,w) + V(F(\phi(\alpha,x),u,w)) \,\big] \\
         & =  \min_{\tilde{u} \in \mathcal{U}} \, \sum_{\tilde{w}} P(\psi(\alpha,\tilde{w})) \times\\
         &\big[\, f(\phi(\alpha,x),\chi(\alpha,\tilde{u}),\psi(\alpha,\tilde{w})) + V(F(\phi(\alpha,x),\chi(\alpha,\tilde{u}),\psi(\alpha,\tilde{w}))) \,\big]\\
         & = \min_{\tilde{u} \in \mathcal{U}} \, \sum_{\tilde{w}} P(\tilde{w})\big[\, f(x,\tilde{u},\tilde{w}) + V(\phi(\alpha,F(x,\tilde{u},\tilde{w})) \,\big] \\
         & = \min_{\tilde u \in \mathcal{U}} \mathbb{E}_{\tilde w}\left[f(x, \tilde u, \tilde w) + V\left(\phi\left(\alpha, F(x, \tilde u, \tilde w)\right)\right)\right]
    \end{aligned}
       \label{eq_A17}
    \end{equation}
\end{enumerate}
This confirms the equality in~\eqref{eq_A7}.

Substituting $V$ with $V^*$ in~\eqref{eq_A7}, and using the fact that $T V^\star = V^\star$, we obtain
\begin{equation}
    \begin{aligned}
        & (TV^\star) \left(\phi(\alpha,x)\right) = \left(T\left(V^\star\left(\phi(\alpha,x)\right)\right)\right)(x) \\
        & \Rightarrow V^\star \left( \phi(\alpha,x) \right) = \left(T\left(V^\star\left(\phi(\alpha,x)\right)\right)\right)(x)
    \end{aligned}
    \label{eq_A19}
\end{equation}
Therefore, $V^\star \left( \phi(\alpha,x) \right)$ is also a fixed point of $T$.
By the minimality of $V^\star$ among nonnegative fixed points
\begin{equation}
    V^\star(x) \;\le\; V^\star \left( \phi(\alpha,x) \right), \quad \forall x \in \mathcal{X}
    \label{eq_A21}
\end{equation}
Applying the same argument for $\alpha^{-1}$ yields 
\begin{equation}
    V^\star(x) \;\le\; V^\star \left( \phi(\alpha^{-1},x) \right), \quad \forall x \in \mathcal{X}
    \label{eq_A22}
\end{equation}
Using the change of variable $x=\phi(\alpha,y)$,~\eqref{eq_A22} becomes
\begin{equation}
    V^\star\left(\phi(\alpha,y)\right) \;\le\; V^\star \left( \phi(\alpha^{-1},\phi(\alpha,y)) \right)
    \label{eq_A23}
\end{equation}
Now, we have
\begin{equation}
    \begin{aligned}
        &  \phi(\alpha^{-1},\phi(\alpha,y)) = \phi\left(\alpha^{-1}*\alpha,y\right) = \phi(e,y) =y
        \end{aligned}
        \label{eq_A24}
\end{equation}
where the first equality follows from~\eqref{eq_A0a}, the second from the group identity $e=\alpha*\alpha^{-1}$, and the third from~\eqref{eq_A0b}.
Substituting \eqref{eq_A24} into \eqref{eq_A23} yields
\begin{equation}
    V^\star(\phi(\alpha,y)) \leq V^\star(y),
    \qquad \forall y \in \mathcal{X}
    \label{eq_A25}
\end{equation}
Substituting $y$ with $x$ in~\eqref{eq_A25} and combining with~\eqref{eq_A21} yields the desired invariance~\eqref{eq_A1}.
\end{proof}
When the conditions of Theorem~\ref{The1} are satisfied and $G$ is a Lie group~\cite{c30}, then Cartan’s moving frame method can be used to construct symmetry–reduced coordinates. Details on how this framework enables dimension reduction are provided in~\cite[Section~2.3]{c21}. 
In the reduced coordinates, the Bellman equation can be solved in a lower–dimensional state space.

\subsection{Dimension Reduction}  \label{App_B}
We now show that the state vector $\begin{bmatrix}h & r\end{bmatrix}^T \in \mathbb{R}^4$ can be reduced to $\begin{bmatrix}e & d & \theta\end{bmatrix} \in \mathbb{R}_+ \times \mathbb{R}_+ \times [0,\pi]$.

Let $ G = \mathrm{SO}(2) $. Each group element $\alpha \in G$ is
parameterized by a rotation angle $\beta \in \mathbb{R}$, i.e.,
\begin{equation}
    \alpha \triangleq R(\beta),
    \qquad
    R(\beta) =
    \begin{bmatrix}
        \cos(\beta) & -\sin(\beta) \\
        \sin(\beta) & \cos(\beta)
    \end{bmatrix}
    \label{eq_B0}
\end{equation}
where $R:\mathbb{R} \rightarrow\mathbb{R}^{2\times2}$. Because the trigonometric functions are $2\pi$-periodic, the angles $\beta$ and $\beta + 2\pi k$ for any $k \in \mathbb{Z}$ represent the same group element.
The group operation is matrix multiplication of rotation matrices, the identity is $R_0=I$, and the inverse element is $\alpha^{-1}=R(-\beta)$.

We define the state as $x=\begin{bmatrix}h & r\end{bmatrix}^T \in \mathcal{X} = \mathbb{R}^4$. The group actions on the state, control, and noise are defined as
\begin{subequations}
    \begin{align}
        &\phi(\alpha,x)=\begin{bmatrix} t \\ t\end{bmatrix} + \begin{bmatrix}\alpha && 0 \\ 0 && \alpha \end{bmatrix}\left(x-\begin{bmatrix} t \\ t\end{bmatrix}\right)  \label{eq_B1a} \\
         &\chi(\alpha,u)=\alpha u \label{eq_B1b}\\
         & \psi(\alpha,w)=\alpha w \label{eq_B1c}
    \end{align}
    \label{eq_B1}
\end{subequations}
We write the system dynamics as
\begin{equation}
    F(x,u,w)=\begin{bmatrix}F_h(h,w) \\ F_r(r,u)\end{bmatrix}
    \label{eq_B2a}
\end{equation}
where $F_h$ and $F_r$ are defined in~\eqref{eq1}. Under the group actions in~\eqref{eq_B1}, the system dynamics in~\eqref{eq_B2a} and the cost in~\eqref{eq3} satisfy the conditions of Theorem~\ref{The1}. Hence, the value function is symmetric, allowing the state space to be reduced. Since the rotation matrix depends on a single parameter, i.e., $\dim G=1$, the reduced state space has dimension $n-\dim G=4-1=3$.

To construct symmetry–reduced coordinates, we apply \emph{Cartan’s moving frame method}~\cite{c21,c32}.
Assume that the group action $\phi(\alpha,x)$ decomposes as $\phi(\alpha,x)=(\phi^a(\alpha,x),\phi^b(\alpha,x))$ with $p$ and $n-p$ components, respectively, where $\phi^a(\alpha,x)$ is invertible.
For some $c$ in the range of $\phi^a(\alpha,x)$, we define a coordinate cross section $\mathcal{C}=\{x:\phi^a(\alpha,x)=c\}$.
Assume that for any $x\in \mathcal{X}$, there is a unique group element $\alpha^* \in G$ such that
\begin{equation}
    \phi(\alpha^*,x) \in \mathcal{C} 
    \label{eq_B2}
\end{equation}
This $\alpha^*$ is called the \emph{moving frame}~\cite{c32}. Using this moving frame, we define the reduced coordinate map $\rho:\mathcal{X} \rightarrow \mathbb{R}^{n-p}$ by
\begin{equation}
    \rho(x) \triangleq\phi^b(\alpha^*,x)
    \label{eq_B3}
\end{equation}
For all $\alpha \in G$ we have $\rho(\phi(\alpha,x))=\rho(x)$,~\cite[Lemma 1]{c32}, meaning that $\rho$ is \emph{invariant} under the group action. Furthermore, the restriction of $\rho$ to $\mathcal{C}$ is injective. For additional details on this method, we refer the reader to~\cite{c21,c32}. 

We now compute the explicit form of the moving frame.
We define the cross-section $\mathcal{C}$ as in~\eqref{eq5a}, and we let $\alpha^*=R(\beta^*)$. Applying~\eqref{eq_B1a} to $r$ and enforcing the cross-section condition~\eqref{eq_B2} yields
\begin{equation}
    t+R(\beta^*)\cdot(r-t) \in \mathcal{C}
    \label{eq_B5}
\end{equation}
Using~\eqref{eq5a}, this yields the conditions
\begin{subequations}
    \begin{align}
        & t^x+\cos(\beta^*)(r-t)^x-\sin(\beta^*)(r-t)^y-t^x>0  \label{eq_B6a} \\
         &t^y+\sin(\beta^*)(r-t)^x+\cos(\beta^*)(r-t)^y-t^y=0  \label{eq_B6b}
    \end{align}
    \label{eq_B6}
\end{subequations}
which leads to
\begin{equation}
        \beta^*(x) = -\operatorname{atan2}\big(-(r-t)^y,(r-t)^x \big)
        \label{eq_B7}
\end{equation}
Substituting $\alpha^*$ into~\eqref{eq_B1a} yields the invariant coordinates
\begin{subequations}
    \begin{align}
        & t^x+\cos(\beta^*)(r-t)^x-\sin(\beta^*)(r-t)^y  \triangleq \rho_1  \label{eq_B8a} \\
         &t^x+\cos(\beta^*)(h-t)^x-\sin(\beta^*)(h-t)^y \triangleq \rho_2  \label{eq_B8b}\\
         &t^y+\sin(\beta^*)(h-t)^x+\cos(\beta^*)(h-t)^y \triangleq \rho_3\label{eq_B8c}
    \end{align}
    \label{eq_B8}
\end{subequations}
Combining~\eqref{eq_B6b}, and~\eqref{eq_B8a} we can simplify the expression for $\rho_1$ in~\eqref{eq_B8a}, and obtain
\begin{equation}
    \rho_1 = t^x + ||r-t||_2
    \label{eq_B9a}
\end{equation}
This yields the following invariant coordinates
\begin{equation}
    \rho(x)=\begin{bmatrix}\rho_1\\ \rho_2\\ \rho_3\end{bmatrix}=\begin{bmatrix}t^x + ||r-t||_2\\ t^x+\cos(\beta^*)(h-t)^x-\sin(\beta^*)(h-t)^y\\ t^y+\sin(\beta^*)(h-t)^x+\cos(\beta^*)(h-t)^y\end{bmatrix}
    \label{eq_B9}
\end{equation}
Note that on the cross-section $\mathcal{C}$, we have $\beta^*=0$, $r^y=t^y$, and $r^x>t^x$, which simplifies the invariant map to
\begin{equation}
    \rho(x)=\begin{bmatrix}r^x \\ h^x\\ h^y\end{bmatrix}
    \label{eq_B10}
\end{equation}
which is an injective map.

We now define a transformation  $\kappa(\rho)$, where $\kappa:\mathbb{R}^3\rightarrow \mathbb{R}_+ \times \mathbb{R}_+ \times [0,\pi]$, such that
\begin{equation}
    \tilde{\rho}=\kappa(\rho)\triangleq\begin{bmatrix}e & d & \theta\end{bmatrix}^T
    \label{eq_B10b}
\end{equation}
Clearly, $e=\rho_1-t^x$. 
 Moreover, from~\eqref{eq_B8a} and~\eqref{eq_B6b}, we obtain
\begin{equation}
    \begin{bmatrix}\rho_1-t^x\\0\end{bmatrix}=\begin{bmatrix}\cos(\beta^*)&-\sin(\beta^*)\\ \sin(\beta^*) & \cos(\beta^*)\end{bmatrix}\begin{bmatrix}(r-t)^x\\(r-t)^y\end{bmatrix} = R(\beta^*)(r-t)
    \label{eq_B12}
\end{equation}
 which implies that
 \begin{equation}
     r-t=R(-\beta^*)\begin{bmatrix}\rho_1-t^x \\ 0\end{bmatrix}
     \label{eq_B13}
 \end{equation}
Similarly, from~\eqref{eq_B8b} and~\eqref{eq_B8c}, we obtain
\begin{equation}
    h-t=R(-\beta^*)\begin{bmatrix}\rho_2-t^x \\ \rho_3-t^y\end{bmatrix}
    \label{eq_B14}
\end{equation}
From Definition~\ref{def1} and using~\eqref{eq_B13} and~\eqref{eq_B14}, we have
\begin{equation}
    \begin{aligned}
       &  d  = ||h-r||_2=||(h-t)-(r-t)||_2= \\
        &= \left|\left|\left(R(-\beta^*)\begin{bmatrix}\rho_2-t^x \\ \rho_3-t^y\end{bmatrix}\right)-\left(R(-\beta^*)\begin{bmatrix}\rho_1-t^x \\ 0\end{bmatrix}\right)\right|\right|_2 \\
        & = \left|\left|\begin{bmatrix}\rho_2-\rho_1 \\ \rho_3-t^y\end{bmatrix}\right|\right|_2 = \sqrt{(\rho_2-\rho_1)^2 + (\rho_3-t^y)^2}
    \end{aligned}
    \label{eq_B15}
\end{equation}
Finally, we have
\begin{equation}
    \begin{aligned}
        & \theta= \arccos\left(\frac{(r-t)^T(h-r)}{||h-r||_2 \cdot||r-t||_2}\right)\\
        &= \arccos\left(\frac{\left(R(-\beta^*) \begin{bmatrix}\rho_1-t^x \\0\end{bmatrix} \right)^T\left(R(-\beta^*)\begin{bmatrix}\rho_2-\rho_1 \\ \rho_3-t^y\end{bmatrix}\right)}{(\rho_1-t^x)\sqrt{(\rho_2-\rho_1)^2 + (\rho_3-t^y)^2}}\right)\\
        &= \arccos\left(\frac{\rho_2-\rho_1}{\sqrt{(\rho_2-\rho_1)^2 + (\rho_3-t^y)^2}}\right)
    \end{aligned}
    \label{eq_B16}
\end{equation}
where the first equality follows from Definition~\ref{def1}, and the second equality follows from~\eqref{eq_B13} and~\eqref{eq_B14}.
Therefore we have
\begin{equation}
    \tilde{\rho}=\kappa(\rho) = \begin{bmatrix}\rho_1-t^x \\  \sqrt{(\rho_2-\rho_1)^2 + (\rho_3-t^y)^2}\\ \arccos\left(\frac{\rho_2-\rho_1}{\sqrt{(\rho_2-\rho_1)^2 + (\rho_3-t^y)^2}}\right)\end{bmatrix}
    \label{eq_B17}
\end{equation}
These quantities provide a complete coordinate system on the reduced state space.

\subsection{Dynamics of Reduced-States}  \label{App_C}

 Let $(h,r)$ be a specific point defined in~(\ref{eq7}).  From Fig.~\ref{Fig2}, and since from~(\ref{eq1}) we know that $r_{+}=F_r(r,u)$ and $h_{+}=F_h(h,w)$, we have
    \begin{equation}
        \begin{aligned}
            &e_{+}=||r_{+}-t||_2 = ||r + u - t ||_2 \\
            &= ||t + e\begin{bmatrix}1 & 0\end{bmatrix}^T +u -t||_2\\
            &=||e\begin{bmatrix}1 & 0\end{bmatrix}^T +u||_2 \\
            & \triangleq F_e(e,u)
        \end{aligned}
        \label{eq_C1} 
    \end{equation}
    where $F_e: \mathbb{R}_+ \times \mathcal{U} \rightarrow \mathbb{R}_+$. To define $\theta_{+}$, let $\nu=r_+-t$ and $\xi=h_+-r_+$, then
    \begin{equation}
        \theta_{+}=\arccos\left(\frac{\nu^T \xi}{||\nu||_2 ||\xi||_2}\right) \triangleq F_\theta(d,e,\theta,u,w)
        \label{eq_C2} 
    \end{equation}
    where $F_\theta:\mathbb{R}_+ \times \mathbb{R}_+ \times [0,\pi] \times \mathcal{U} \times \mathcal{W} \rightarrow [0,\pi]$, and where
    \begin{equation}
         \nu = \begin{bmatrix}e + u^x \\ u^y\end{bmatrix}, \quad
         \xi = \begin{bmatrix}d \cos(\theta) + w^x -u^x \\ d \sin(\theta) + w^y - u^y\end{bmatrix}
        \label{eq_C3} 
    \end{equation}
    This means that $\theta_{+}$ is a function of $d$, $e$, $\theta$, $u$, and $w$. Note that in the above equations, $u=\begin{bmatrix}u^x \\ u^y\end{bmatrix}$, and $w=\begin{bmatrix}w^x \\ w^y\end{bmatrix}$.
    Moreover, $d_+$ is defined as
    \begin{equation}
        \begin{aligned}
            &d_{+}= ||h_{+}-r_{+}||_2\\
            &=||d \begin{bmatrix}\cos(\theta) & \sin(\theta)\end{bmatrix}^T +w-u||_2 \\&\triangleq F_d(d,\theta,u,w)
        \end{aligned}
        \label{eq_C4} 
    \end{equation}
where $F_d:\mathbb{R}_+ \times [0,\pi] \times \mathcal{U} \times \mathcal{W} \rightarrow \mathbb{R}_+$. Hence, from~(\ref{eq_C1}),~(\ref{eq_C2}), and~(\ref{eq_C4}), for a specific pair $(h,r)$ defined as~(\ref{eq7}), $d_+$, $e_+$, and $\theta_+$, can be expressed as functions of $d$, $e$, $\theta$, $u$, and $w$.

\else 

In this appendix, we extend the symmetry–reduction framework of~\cite{c21}, originally developed for finite–horizon stochastic optimal control, to the SSP setting.

\fi

\section*{Acknowledgment}

This work was partially supported by the Wallenberg AI, Autonomous Systems and Software Program (WASP) funded by the Knut and Alice Wallenberg Foundation.



%

\end{document}